\documentclass[pra,twocolumn,amsmath,amssymb]{revtex4-1}

\usepackage{epsfig,amsmath}
\usepackage{subfigure}
\usepackage{graphicx}
\usepackage{dcolumn}
\usepackage{stmaryrd}
\usepackage{mathrsfs}
\usepackage{pifont}
\usepackage{amsthm}
\usepackage{amssymb}
\usepackage{bm}
\usepackage{latexsym}
\usepackage[colorlinks=true,linkcolor=blue,citecolor=blue]{hyperref}
\usepackage{color}

\theoremstyle{plain}
\newtheorem{theorem}{Theorem}

\newcommand{\la}{\langle}
\newcommand{\ra}{\rangle}
\newcommand{\lam}{\lambda}
\newcommand{\ti}{\tilde}

\newcommand{\da}{\dagger}
\newcommand{\De}{\Delta}
\newcommand{\al}{\alpha}
\newcommand{\si}{\sigma}

\newcommand{\om}{\omega}

\newcommand{\non}{\nonumber}

\newcommand{\ep}{\epsilon}

\newcommand{\blue}[1]{\textcolor{blue}{#1}}

\def\oe#1{{ Opt.\ Express} {\bf#1}}

\def\jpb#1{{ J.\ Phys.\ B} {\bf#1}}
\def\jpa#1{{ J.\ Phys.\ A} {\bf#1}}
\def\pra#1{{ Phys.\ Rev. A\/} {\bf#1}}
\def\prb#1{{ Phys.\ Rev. B\/} {\bf#1}}

\def\prl#1{{ Phys.\ Rev.\ Lett.} {\bf#1}}

\def\sci#1{{ Science} {\bf#1}}

\def\pla#1{{ Phys.\ Lett. A\/} {\bf#1}}
\def\rmp#1{{ Rev. \ Mod. \ Phys.} {\bf#1}}
\def\nat#1{{ Nature} {\bf#1}}

\def\njp#1{{ New. J. \ Phys.} {\bf#1}}

\begin{document}

\title{Fundamental Limitation on Cooling under Classical Noise}

\author{Jun Jing$^{1,2}$, Ravindra W.  Chhajlany$^{3}$, and Lian-Ao Wu$^{2,4}$}\thanks{Author to whom any correspondence should be addressed. Email address: lianao.wu@ehu.es}

\affiliation{$^{1}$Institute of Atomic and Molecular Physics and Jilin Provincial Key Laboratory of Applied Atomic and Molecular Spectroscopy, Jilin University, Changchun 130012, Jilin, China \\ $^{2}$Department of Theoretical Physics and History of Science, The Basque Country University (EHU/UPV), PO Box 644, 48080 Bilbao, Spain \\ $^{3}$Faculty of Physics, Adam Mickiewicz University, Umultowska 85, 61-614 Pozna{\'n}, Poland \\$^{4}$Ikerbasque, Basque Foundation for Science, 48011 Bilbao, Spain}

\date{\today}

\begin{abstract}
We prove a general theorem that the action of arbitrary classical noise or random unitary channels can {\it not} increase the maximum population of {\it any} eigenstate of an open quantum system, assuming initial system-environment factorization. Such factorization is the conventional starting point for descriptions of open system dynamics. In particular, our theorem implies that a system can {\it not} be ideally cooled down unless it is initially prepared as a pure state. The resultant inequality rigorously constrains the possibility of cooling the system solely through temporal manipulation, {\it i.e.,} dynamical control over the system Hamiltonian without resorting to measurement based cooling methods.
\end{abstract}


\maketitle

\section{Introduction}

Cooling and, more generally, pure-state preparation~\cite{gscool,sideband,gscool2,YuasaPRL03,LiPRB11} of a microscopic or mesoscopic open system (small thermal object)~\cite{BreuerBook} is of paramount importance to many intriguing quantum technologies and engineering of low temperature quantum phases, in general. Examples of applications include quantum simulations of many body physics~\cite{Lewenstein-book} on a variety of platforms such as cold atoms and molecules, trapped ions and nanophotonic systems. Similarly, quantum computers~\cite{Zoller95}, the promising quantum adiabatic computing (QAC) paradigm~\cite{AQC1,AQC2,AQC3,Dwave}, quantum communication ~\cite{Zeilinger}, dynamically enhanced nuclear polarization~\cite{DNP1,DNP2} and quantum metrology are merely a few prominent examples of contemporary applications where significant control over quantum states needs to be exercised. More specifically, for perfect realization of quantum logic operations, qubits initially need to be cooled down to the ground state of motion prior to coherent manipulation~\cite{WisemanBook}. Any cooling scheme, e.g., bang-bang cooling~\cite{bang}, single-shot state-swapping cooling~\cite{shot}, and sideband cooling~\cite{sb1,sb2,sb3,sb4}, cannot be performed when the system is isolated~\cite{BreuerBook}.

The quantum adiabatic computation is an interesting paradigm for universal quantum computation. Here the solution to a hard problem is encoded in the ground state of a many body Hamiltonian, {\it i.e.,} the computer. To reach the solution, the computer is initialized to the ground state of some Hamiltonian that can be easily prepared. The initial ground state is then transported adiabatically~\cite{Berry1,Berry2,Zanardi,Lidar1,Lidar2,Lidar3} to the target ground state encoding the solution. In principle, adiabaticity suppresses errors in the preparation of the final ground state by overcoming the problem of energy relaxation~\cite{GsC1,GsC2,GsC3,JocobsBook} as the system at all times is kept in the ground state of the instantaneous Hamiltonian during evolution. However, the rate of change of the Hamiltonian control parameters, and so the protocol's running time, scales inversely with the square of the spectral gap to the lowest excitations. In practice, the system is always excited during the protocol, most seriously because one must generically move through regions in parameter space where the gap is very small or closes completely. Apart from this, the system is never truly isolated from its environment, which also results in excitations. One solution to overcome this problem is to combine the quasi-adiabatic evolution with active cooling to suppress such errors generated via excitations during the running of the protocol.

\blue{Addressing feasibility of such schemes motivates the search for a better understanding of the description of cooling effects in open system dynamics, in particular as described below, with classical noise. Cooling setups consist of a small target object (e.g., a mechanical resonator) and an ancillary system (e.g., a qubit) as the {\em entire} system  which is embedded in a quantum environment. The three entities seem to be equally crucial in a cooling process. The dynamics of the {\em entire} system is supposed to be governed by the conventional quantum Markovian master equation~\cite{note}.  Recently it was interestingly shown that fully quantum-mechanical models under the Born-Markov approximation may be mapped sometimes to a quantum system under classical noise~\cite{quantumclassical}. On the one hand, this inspires the question about whether the Born-Markov approximation generically allows an arbitrary full quantum bath to be equivalent to some corresponding classical noise. If so, can the {\em entire} system be cooled by the equivalent classical noise via the quantum Markovian master equation? On the other hand, is it possible that this specific classical-quantum equivalence is fake due to the Born-Markov approximation? These issues are too difficult to be solved in generality with any known analytical and numerical techniques. Therefore, setting up strict quantum bounds is absolutely necessary in studying cooling problems. An example is the recently proposed ``counterintuitive'' protocols as cooling by heating~\cite{cbh1,cbh2} with the help of ``incoherent thermal quantum noise'' from a high-temperature bath, where we know "quantum noise" could be equivalent to a corresponding classical noise as long as the Born-Markov approximation is used. In Ref.~\cite{cbh1}, the authors consider an ancillary system of two optical modes coupled to a mechanical degree of freedom and find the mechanical oscillator can be cooled down to an extent by heating one of the optical modes {\it i.e.}, increasing its thermal state population. It is thus interesting to consider the constraints on the types of processes that can be realized under restricted operations such as evolution under classical noise, and to unambiguously identify the origin of such counter-intuitive effects. }

Interestingly, it has been recently proved that exact ground state cooling is forbidden when one assumes factorization of the initial state of the system from the bath state~\cite{WuSR13}. This is remarkable since initial system-bath product state factorization is a common condition adopted in the derivation of master equations or Kraus operator representations describing open system dynamics~\cite{BreuerBook,Kraus2}. Here we ask the less stringent question of whether approximate cooling --- understood as increasing the ground state population --- can be achieved under such system-bath factorization. We consider the case of coupling the system to classical environmental noise which can be thought of as stochastically affecting the control parameters of the system Hamiltonian. We find that under such conditions, even approximate cooling is impossible.

\section{No-go theorem for cooling an open system under classical noise}

Classical noise~\cite{GardinerBook,SpinCnoise} corresponds to a special form of the system-environment interaction Hamiltonian $H_I=\sum_jA_jB_j$, where $A_j$'s are Hermitian operators in the Hilbert space of the system and $B_j$'s are environmental operators, when it can be semi-classically approximated by $H_I=\sum_jA_j\la B_j\ra$, where the $\la B_j\ra$'s are now c-numbers (instead of operators) determined by the random states of the environment, time-independent or time-dependent. This yields a stochastic Hamiltonian acting on the system, $H_\lam=H_0(t)+H_I(\lam)$~\cite{note2}, where the system bare Hamiltonian $H_0(t)$ is in general time-dependent which might take into account the general possibility of control via external parameters, and $\lam\equiv\{\la B_j\ra\}$ represents a random parameter characterizing a particular realization of the system evolution.

The system evolution determined by $H_\lambda$ corresponding to a particular realization of the stochastic environmental parameters is unitary, $\rho_i \mapsto K(\lam)\rho_{\rm i}K(\lam)^\da$, where $K(\lam)=K_{t\geq0}(\lam)$ is defined as a propagator $\mathcal{T}e^{-i\int_0^tdsH_\lam}$ with $\hbar\equiv1$. Here the time-ordering symbol $\mathcal{T}$ accommodates the general situations, in which $H_{\lam}$ can be time-dependent. A particular evolution given above is not enough to obtain the real evolution of the open system under noise. The configuration of environmental variables is in general unknown and may be assumed to be described by a probability distribution $|p_\lambda|^2$. The final (evolved) state of the system must be the average over all possible unitary evolutions of the type just described and therefore
\begin{equation}\label{Kraus}
\rho_{\rm f}=\sum_\lam|p_\lam|^2K(\lam)\rho_{\rm i}K(\lam)^\da,
\end{equation}
with $\sum_\lam|p_\lam|^2=1$. This is a general expression independent of the details of the  system Hamiltonian $H_0(t)$. The above equation is a special case of the Kraus operator representation of open system dynamics, which we briefly recap in appendix~\ref{KR}.

With this definition, we introduce the following no-go theorem:
\begin{theorem}\label{limit}
For any quantum operation process describing uncertainty-induced decoherence defined by Eq.~(\ref{Kraus}), the system can not be completely transferred into a pure state unless it is prepared as one initially.
\end{theorem}

\begin{proof}
Without loss of generality, both initial and final states $\rho_{\rm i}$ and $\rho_{\rm f}$ can {\it always} be written in their respective diagonal forms as $P={\rm diag}[P_1, P_2, \cdots]$ and $Q={\rm diag}[Q_1, Q_2, \cdots]$, respectively. In other words, suppose $\rho_{\rm i}=W^\da PW$ and $\rho_{\rm f}=V^\da QV$, where $W$ and $V$ are unitary operators that diagonalize the initial and final states of the system, respectively, then one can redefine $K(\lam)$ as $VK(\lam)W^\da$. Equation~(\ref{Kraus}) can be rewritten as $Q_m=\sum_{\lam,n}|p_\lam|^2|K_{mn}(\lam)|^2P_n$. Throughout this work, we consider the eigenstate populations to be ordered in decreasing order: $1\geq P_1\geq P_2\geq P_3\cdots$, and $1\geq Q_1\geq Q_2\geq Q_3\cdots$. Using this as well as the normalization conditions $\sum_n|K_{mn}(\lam)|^2=1$ and $\sum_\lam|p_\lam|^2=1$, the following inequality holds,
\begin{equation}\label{cool}
Q_m\leq P_1\sum_{\lam,n}|p_\lam|^2|K_{mn}(\lam)|^2=P_1,
\end{equation}
which implies, in particular, that $Q_1\leq P_1$, {\it i.e.}, the maximally occupied eigenstate of the final state cannot have a larger population than that of the corresponding initial state. So if $P_1<1$ for a mixed initial state, then $Q_1<1$ and $\rho_f$ can not be a pure state.
\end{proof}

At a microscopic level, the strong notion of cooling corresponds to demanding that the population of the ground state increases during the cooling process. In particular, this is in strict agreement with the phenomenology of cooling when both the initial and final states are Gibbs-ensemble equilibrium states. Ideal cooling is attained when the final state is the ground state of the system $|0\ra\la0|$. \textbf{Theorem}~\ref{limit} shows that cooling even in an approximate sense, {\it i.e.,} increasing the population of the ground state by an arbitrarily small amount is impossible under solely classical decoherence with the implicit constraint of initial system-environment factorization.

In order to turn this microscopic picture into a macroscopic one, consider the initial and final states to be thermal Gibbs states at two respective temperatures $T_i$ and $T_f$. It is simplest to consider a two level system, although the same arguments hold for a system with many energy levels. For such a system, the initial and final temperatures are given by
\begin{equation}\label{tem}
T_{\rm i}=\frac{\om_i}{\kappa_B}\ln\frac{P_1}{1-P_1}, \quad T_{\rm f}=\frac{\om_f}{\kappa_B}\ln\frac{Q_1}{1-Q_1},
\end{equation}
where the initial and final energy spacings of the system are $\om_{i, f}$, and the initial and final populations of the ground state are $P_1$ and $Q_1$, respectively. Since $Q_1\leq P_1$ due to our no-go theorem, $T_{\rm f}\leq \frac{\om_i}{\om_f}T_{\rm i}$. Hence to surely cool the system, one must impose $\om_f>\om_i$ which can only be achieved by doing work on the system resulting in the changing of the spectral properties of the system. However, we are here considering the conventional approach to cooling (in particular to the ground state of) a Hamiltonian which is the same at the initial and final instants. So we have that $\om_f=\om_i$ and in terms of temperature one obtains $T_f \geq T_i$, {\it i.e.,} the temperature of the system cannot be reduced.

Classical noise is expected to be the most common source of disturbance in many open quantum systems, such as telegraph noise in ion traps~\cite{CN1} or $1/f$ noise in solid systems~\cite{CN2}. We note that the specific (implicit) time-dependence of Kraus operators in Eq.~(\ref{Kraus}) depends on the statistical features of the classical noise for the system. When the noise correlation function at different times is proportional to the $\delta$-function, the dynamics of the density matrix is equivalent to that described by the conventional Lindblad master equation. In literatures, this corresponds to white noise or Markovian noise. Otherwise, the noise is colored or non-Markovian. We emphasize here that independently of these characteristics, classical noise is always characterised by a group of unitary transformation $K(\lam)$'s and our discussion holds in general.

\section{Consequences of the no-go theorem}

The classical noisy Hamiltonian is realised by adding stochastic processes to the system's Hamiltonian. It is meaningful in various physical situations, where ambient noise is assumed to be additive under a certain probability distribution. The final state or dynamics of the system is then obtained by an ensemble average of the form~\eqref{Kraus}. For instance, it is sufficient to treat the hyperfine interaction between the electron spin and environmental nuclear spin as a classical entity which describes the inhomogeneous broadening process of the electron spin confined in a quantum dot~\cite{HB}. We explicitly illustrate the general result (\textbf{Theorem}~\ref{limit}) through three concrete examples of the action of classical noise on popularly studied systems in quantum control theory.

\subsection{A single two-level system}

Suppose the evolution of a two-level atomic system, initially in the state ${\rm diag}([P_1, P_2])$ (assuming $P_1\geq P_2$ without loss of generality), can be described by the Kraus representation
\begin{eqnarray}\label{Kraus2}
\ti{\rho}_{\rm f}=\sum_{k=1}^N\lam_kE_k\rho_{\rm i}E_k^\da, \quad E_k=\left(\begin{array}{cc}\cos\theta_k & i\sin\theta_k \\ i\sin\theta_k & \cos\theta_k \end{array}\right)
\end{eqnarray}
where $\lam_k\geq0$, $\sum_{k=1}^N\lam_k=1$ and $\theta_k$ represents the $k$-th realization of stochastic disturbance. $E_k$ describes a general and random unitary transformation characterized by $\theta_k$ for arbitrary two-level systems. Physically, this random channel can describe an electron spin interacting with a magnetic field subject to small stochastic fluctuations along the $z$- and $x$-directions. One can keep in mind that Eq.~(\ref{Kraus2}) is a particular realization of Eq.~(\ref{Kraus}), so that here $\theta_k$ and $\lam_k$ correspond to $\lam$ and $|p_\lam|^2$, respectively. The diagonal representation of the final state $\ti{\rho}_{\rm f}$ is given by the populations $Q_1$ and $Q_2$ which can be expressed by $(1\pm X)/2$, where $X$ can be measured by the auxiliary quantity $Y\equiv X^2/(2P_1-1)^2$ below. We have
\begin{eqnarray*}
Y&=&\left(\sum_{k=1}^N\lam_k\cos2\theta_k\right)^2
+\left(\sum_{k=1}^N\lam_k\sin2\theta_k\right)^2 \\
&\leq& \sum_{k=1}^N\lam_k^2\left(\sum_{k=1}^N\cos^2 2\theta_k+\sum_{k=1}^N\sin^22\theta_k\right)=N\sum_{k=1}^N\lam^2_k.
\end{eqnarray*}
According to the condition of Cauchy$-$Schwarz inequality, the maximal value of $Y$ is attained iff $\lam_k/\cos2\theta_k=c_1$ and simultaneously $\lam_k/\sin2\theta_k=c_2$, where $c_1$ and $c_2$ are constant numbers independent of $k$. In this case, $\lam_k$ must be taken as $1/N$, which is independent on $k$, in order for $Y$ as well as $X$ and $Q_1$ or $Q_2$ to achieve the maximum value. Then we have found $\max\{Y\}=1$. So that $X\leq 2P_1-1$, and $Q_1, Q_2\leq P_1$ as advertised.

\subsection{A mechanical resonator (MR)}

Consider a doubly-clamped mechanical resonator embedded in a flux-qubit circuit (serving as an auxiliary qubit), which is composed of superconducting loops with Josephson junctions~\cite{MR1,MR2}. An in-plane magnetic field $B$ induces qubit-MR coupling via a Lorentz force. Upon tuning the tunneling amplitude $\De$ between the two persistent current states $|\downarrow\ra$ and $|\uparrow\ra$ to be near resonant with the qubit-MR frequency $\om_m$, but much larger than the qubit-MR coupling constant $g$, it is proper to approximate the Hamiltonian as $H=\om_ma^\da a+\De\si_z/2+g(a\si_++a^\da a_-)$. Here the $\si_{z,+,-}$ are the Pauli operators in the new basis of ground and excited states, $|g/e\ra\equiv(|\downarrow\ra-/+|\uparrow\ra)/\sqrt{2}$. The Hamiltonian can therefore be diagonalized into  $H=-\De/2|0,g\ra\la0,g|+\sum_{n,s=\pm}\ep_n^s|ns\ra\la ns|$, where the dressed eigenstates are $|0g\ra$ and  $|n+\ra=\cos\al_n|n-1,e\ra+\sin\al_n|n,g\ra$, and $|n-\ra=\sin\al_n|n-1,e\ra-\cos\al_n|n,g\ra$, with $n\geq1$. Here $\tan2\al_n=2g\sqrt{n}/(\De-\om_m)$. The MR system is assumed to be under the influence of a classical noisy Hamiltonian $H_I(\lam_k)=\theta_ka^\da a$ induced by the random pressure from the phonons of mechanical oscillator. Thus, from a diagonal state $\rho_{\rm i}=\rho_{n\rm i}^{\bigoplus n}$, the final state can be obtained by $\ti{\rho}_{\rm f}=\ti{\rho}_{n\rm f}^{\bigoplus n}$, which is in a block-diagonal formation. For $n=0$, $\ti{\rho}_{0\rm f}=\ti{\rho}_{0\rm i}$. When $n\geq1$,
\begin{equation}
\ti{\rho}_{n\rm f}=\sum_{k=1}^N\lam_kE_{k,n}\rho_{n \rm i}E_{k,n}^{\da},
\end{equation}
where the four elements of each $2\times2$ blocks $E_{k,n}$ are
\begin{eqnarray} \non
E_{k,n}^{11}&=&e^{-i\theta_k(n-1)}
(\cos^2\al_n+e^{-i\theta_k}\sin^2\al_n),  \\  \non E_{k,n}^{12}&=&e^{-i\theta_k(n-1)}
(1-e^{-i\theta_k})\cos^2\al_n\sin^2\al_n, \\   \non
E_{k,n}^{21}&=&E_{k,n}^{12}, \\
E_{k,n}^{22}&=&e^{-i\theta_k(n-1)}
(\sin^2\al_n+e^{-i\theta_k}\cos^2\al_n).
\end{eqnarray}
Suppose in each $2\times2$ block of $\rho_{\rm i}$, $P_{1n}\geq P_{2n}$ and $P_{1n}+P_{2n}=P_n$, $\ti{\rho}_{0\rm i}+\sum_{n\geq1}P_n=1$. After a straightforward derivation, the two populations $Q_{1n}$ and $Q_{2n}$ in the $n$-th block of $\rho_{\rm f}$ are evaluated as $(P_n\pm X_n)/2$. The auxiliary quantity $Y_n\equiv X_n^2/(P_{1n}-P_{2n})^2$ that measures the difference between $Q_{1n}$ and $Q_{2n}$ is found to be
\begin{equation*}
Y_n=\sum_{j=1}^3\left(\sum_{k=1}^N\lam_k\mu^{(j)}_{kn}\right)^2 \leq\left(\sum_{k=1}^N\lam_k^2\right)\sum_{j=1}^3\sum_{k=1}^N
\left(\mu^{(j)}_{kn}\right)^2,
\end{equation*}
where $\mu^{(1)}_{kn}\equiv(1-\cos\theta_k)\sin4\al_n$, $\mu^{(2)}_{kn}\equiv\sin\theta_k\sin2\al_n$, and $\mu^{(3)}_{kn}\equiv1-2\sin^2\frac{\theta_k}{2}\sin^22\al_n$. Due to the fact that $\sum_{j=1}^3(\mu^{(j)}_{kn})^2=1$ and the Cauchy$-$Schwarz inequality, $Y_n$ as well as $X_n$ and $Q_{1n}$ or $Q_{2n}$ achieves the maximum value also iff $\lam_k=1/N$. So that $|Q_{1n}-Q_{2n}|\leq|P_{1n}-P_{2n}|$ and then $Q_{1n}, Q_{2n}\leq P_{1n}$.

\subsection{A three-level system}

Consider the stimulated Raman adiabatic passage (STIRAP)~\cite{STIRAP} in a three-level atomic system, which targets the perfect population transition between $|0\ra$ and $|2\ra$ without disturbing the quasistable state $|1\ra$. The system can be adiabatically evolved from $|0\ra$ to $-\sin\theta|2\ra+\cos\theta|0\ra$ under a time evolution operator with two parameters~\cite{IE} as,
\begin{equation*}
U(\theta,\alpha)=\left(\begin{array}{ccc}
\cos\theta\cos\al & \cos\theta\sin\al & -\sin\theta \\
 -\sin\al & \cos\al & 0 \\
\sin\theta\cos\al & \sin\theta\sin\al & \cos\theta
\end{array}\right).
\end{equation*}
Physically, this is a standard time-evolution operator used to inversely engineer the STIRAP process. Classical noise would fluctuate the parameters in $U(\theta,\alpha)$. One can then let $\theta\rightarrow\theta_k$ or $\alpha\rightarrow\alpha_k$, meaning the parameters are no longer stable but fluctuate due to noise, and then random unitary transformation $E_k=U(\theta_k,\alpha)$ or $E_k=U(\theta,\alpha_k)$ is applied to analysis the effect from two types of classical noise channels on the system.

\begin{figure}[htbp]
\centering
\includegraphics[width=3in]{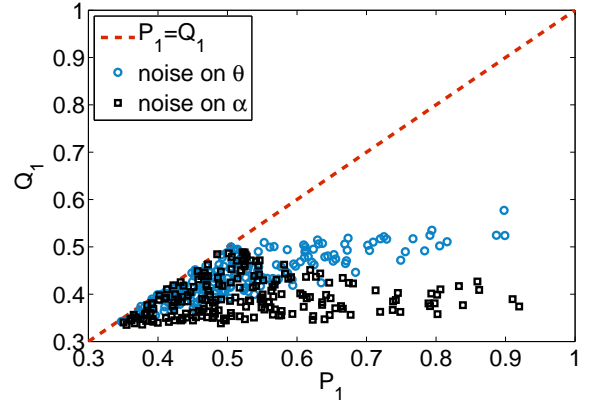}
\caption{Numerical comparison of the final and initial ground state populations $Q_1$ vs. $P_1$ for a three-level system subject to noisy stimulated adiabatic passage. Considered are randomly generated initial states, noise distributions $\lam_k$'s, and classical noise affecting the parameter $\theta$ (circles, with $\alpha$ fixed as $\cos\alpha=\sqrt{1/3}$) and $\alpha$ (squares, with $\theta$ fixed as $\cos\theta=\sqrt{7/10}$). }\label{cool3L}
\end{figure}

Substituting the Kraus operator $E_k$ into Eq.~(\ref{Kraus2}), we calculate $\ti{\rho}_{\rm f}$ by random chosen probability $\lam_k$ and then obtain the diagonalized form of $\rho_{\rm f}$. In Fig.~\ref{cool3L}, we compare the numerical results of $200$ random generated $P_1$ and the corresponding $Q_1$ obtained in different configurations of classical noise. Each $Q_1$ is ensemble averaged with $N=100$ noisy time-evolution operators $E_k$ [see Eq.~(\ref{Kraus2})]. We can see for both $\theta$ and $\alpha$, classical noise prevents $Q_1$ from exceeding $P_1$. For $P_1\lesssim0.5$, certain realization of noise allow $Q_1$ to be equal to $P_1$. On the other hand, for $P_1\gtrsim0.5$, the difference between $Q_1-P_1$ roughly increases with $P_1$.

\section{Extension to  quantum channels}

Our theory on classical noise can be extended to certain channels of quantum systems. One can provide a sufficient condition on the general Kraus representation of the action of noise for even approximate cooling to be impossible. Let $\hat{Q}=\sum_{\lam}E_\lam\hat{P}E_\lam^\da$, where $\hat{P}$ and $\hat{Q}$ indicate the diagonalized $\rho_i$ and $\rho_f$, respectively. The individual populations therefore satisfy
\begin{eqnarray}\non
Q_m&=&\la m|\hat{Q}|m\ra=\sum_{\lam,n}\la m|E_\lam|n\ra\la n|\hat{P}|n\ra\la n|E_\lam^\da|m\ra\\ &=&
\sum_{\lam,n}|\la m|E_\lam|n\ra|^2P_n\leq P_1\sum_{\lam,n}|\la m|E_\lam|n\ra|^2,
\end{eqnarray}
where $P_n=\la n|\hat{P}|n\ra$. Now, we notice that if for arbitrary $m$, $\sum_{\lam,n}|\la m|E_\lam|n\ra|^2\leq1$, then $Q_m\leq P_1$, which yields the claimed result $Q_1\leq P_1$. As an illustration for a two-level system, see the following list of Kraus operators that satisfy the normalization condition $\sum_\lam E_\lam^\da E_\lam=1$ and are commonly used in the theory of error-correction. The first example is the bit flip channel flips the state of a qubit from $|0\ra$ to $|1\ra$ (and vice versa) with probability $1-p$. It has operation elements:
\begin{equation*}
E_1=\left(\begin{array}{cc}\sqrt{p} & 0 \\ 0 & \sqrt{p}  \end{array}\right), \quad
E_2=\left(\begin{array}{cc}0 & \sqrt{1-p}  \\ \sqrt{1-p}  & 0 \end{array}\right).
\end{equation*}
The second one is the phase flip channel with operation elements
\begin{equation*}
E_1=\left(\begin{array}{cc}\sqrt{p} & 0 \\ 0 & \sqrt{p}  \end{array}\right), \quad
E_2=\left(\begin{array}{cc}\sqrt{1-p} & 0 \\ 0 & -\sqrt{1-p}  \end{array}\right).
\end{equation*}
The third one, the bit-phase flip channel is characterized by
\begin{equation*}
E_1=\left(\begin{array}{cc}\sqrt{p} & 0 \\ 0 & \sqrt{p}  \end{array}\right), \quad
E_2=\left(\begin{array}{cc}0 & -i\sqrt{1-p}  \\ i\sqrt{1-p}  & 0 \end{array}\right).
\end{equation*}
The fourth one is the depolarizing channel, which is an important type of quantum noise. To lead the qubit to be depolarized with probability $p$, the Kraus operation elements are chosen as $E_1=\sqrt{1-3p/4}\mathcal{I}$, $E_2=\sqrt{p}\si_x/2$, $E_3=\sqrt{p}\si_y/2$, and $E_4=\sqrt{p}\si_z/2$, where $\mathcal{I}$ and $\si$'s are identity operator and Pauli operators, respectively. For all of the above quantum operations, one can check that $\sum_{\lam,n}|\la m|E_\lam|n\ra|^2=1$ for either $m=1$ or $m=2$, so that $Q_1=P_1$, which means no cooling can take place. 

\section{Discussion and Conclusion}

One could hope that cooling might be realized with minimal resources, such as under the influence of classical noise. Here, we have shown that even approximate cooling to the ground state is impossible under such conditions assuming no work is done on the system. This may be viewed as a fundamental limitation in the theory of open system dynamics. In principle it means that cooling methods described in a standard way, using initial system-environment state factorization, must include quantum noise (a finite-temperature quantum environment) and/or feedback mechanisms based on relevant measurements allowing extraction of information from the system. These protocols cannot be generally described in terms of the random channel~\eqref{Kraus}. For example, the dynamics induced by projective measurements is generated by a non-Hermitian operator~\cite{YuasaPRL03,LiPRB11}. It is noteworthy to state here that our result is independent of system dimension, as illustrated by the examples presented. As a direct application, our no-go theorem implies, e.g., that the exciton energy transfer in light-harvesting complexes at room temperature~\cite{bio} is assisted by non-classicality of the molecular vibrations. The discussion of whether or not this process bears quantum features was an interesting problem brought up in the Ref.~\cite{bio} and was found to be a non-trivial problem to resolve. 

We believe that it is important to understand  constraints on cooling mechanisms under various types of system-environment couplings. We mention here that not all systems can even be cooled quantum-mechanically. In fact the cooling rate and the lowest achievable steady temperatures for optomechanical system~\cite{sideband} and micromechanical system~\cite{lasercool} are determined by the resonator's quantum fluctuations (photon shot noise). In certain regime when the frequency of the mechanical system is smaller than the decay rate of the cavity, the cooling might fail even with asymmetry in the noise spectrum. We recall that spectral asymmetry -- which means that excitation and dexcitation of the system via the environment are inequivalent processes -- is the usual condition required to achieve cooling of a quantum system. Similarly, the ground-state cooling is impossible for initial phonon numbers larger than mechanical quality factor. Experimentally~\cite{lasercool}, as the cooling laser power is increased, it is shown that the system will arrive at the quantum backaction limit, with equal sideband heights as the mechanical resonator comes into equilibrium with the optical bath. 

Finally, we note that our results offer a different perspective to that provided by existing specific limitations on cooling protocols, such as the recently shown limitation on the amount of steady state entanglement that can be generated when subjecting the system only to local dissipation~\cite{steady}.  The latter limitation means that systems with sufficiently entangled states cannot be cooled to the ground state under local dissipative processes.

\appendix
\section{Kraus representation emerging in the open quantum system dynamics}\label{KR}

The time evolution of an open quantum system plus its environment is governed by some joint unitary propagator $U(t)$, which describes the evolution of the joint density matrix $\rho(0)$ from time $t=0$ to time $t$, {\it i.e.,} $\rho(t)=U(t)\rho(0)U^\da(t)$. The reduced dynamics of the system is commonly described in terms of a master equation, or equivalently is mathematically described as a quantum channel which technically is a completely positive trace preserving map. The reduced density matrix describing the system can be represented by $\rho_{\rm f}=\sum_k\la e_k|U\rho(0)U^\da|e_k\ra$, where $|e_k\ra$ is an orthonormal basis for the environment. It is convenient to choose the environmental basis to be the one that diagonalizes the state of the environment at $t=0$,  $\rho_B(0)=\ep_l|e_l\ra\la e_l|$. In general, assumption of initial independence of system and environment, {\it i.e.}, the factorized form $\rho(0)=\rho_{\rm i}\otimes\rho_B(0)$, is necessary for the derivation of the master equation (or quantum channel) describing the dynamics. Under this condition a sum-up representation of the open system evolution can always be expressed in the form $\rho_{\rm f}=\sum_{kl}\ep_lE_{kl}\rho_{\rm i}E_{kl}^\da$, where $E_{kl}\equiv\la e_k|U(t)|e_l\ra$'s are so-called Kraus operators which satisfy the normalization condition $\sum_{kl}\ep_lE_{kl}^\da E_{kl}=1$.

\acknowledgments
We acknowledge grant support from the Basque Government (grant IT472-10), the Spanish MICINN (No. FIS2012-36673-C03-03), the National Science Foundation of China No. 11575071 and Science and Technology Development Program of Jilin Province of China (20150519021JH).


\begin{thebibliography}{99}
\bibitem{note2}
\blue{{\bf The PRA referee's comment:} ``but let me note that the time-ordering is needed not only when $H_{\lambda}$ is time dependent, but also when $H_\lambda$ does not commute with $H_0$ (i.e. in all interesting cases). Since $H_0$ is not present in the evolution operator explicitly I am assuming that the Authors are working in the interaction picture, so that $H_\lambda$ under the exponent is in fact a time-dependent interaction picture operator given by $e^{iH_0 t} H_\lambda e^{-iH_0 t}$ (for time-independent $H_0$ for simplicity).'' {\bf Our reply}: This is also incorrect. We use $H_\lambda$ to represent the total Hamiltonian such that the evolution operator of $H_\lambda$ is time-ordering. We never use interacting picture in the whole manuscript and the referee¡¯s own assumption is not relevant. Note that we focus on classical noise so that $H_\lambda$ is the total Hamiltonian.}
\end{thebibliography}
\end{document}